\documentclass[letterpaper,twocolumn,10pt]{article}
\makeatletter

\if@twocolumn\else\input twocolumn.sty\fi
\usepackage{times}  

\usepackage[T1]{fontenc}
\usepackage[utf8]{inputenc}

\usepackage[kerning=false,spacing=false]{microtype} 
\usepackage{flushend} 

\usepackage{cite}               
\usepackage{url}                

\usepackage{xcolor}             
\usepackage[]{hyperref}         
\hypersetup{
  colorlinks,
  linkcolor={blue!60!black},
  citecolor={blue!60!black},
  urlcolor={blue!60!black}
}

\def\maketitle{\par
 \begingroup
   \renewcommand\thefootnote{\fnsymbol{footnote}}%
   \def\@makefnmark{\hbox to\z@{$\m@th^{\@thefnmark}$\hss}}%
    \long\def\@makefntext##1{\parindent 1em\noindent
            \hbox to1.8em{\hss$\m@th^{\@thefnmark}$}##1}%
   \if@twocolumn
     \twocolumn[\@maketitle]%
     \else \newpage
     \global\@topnum\z@
     \@maketitle \fi\@thanks
 \endgroup
 \setcounter{footnote}{0}%
 \let\maketitle\relax
 \let\@maketitle\relax
 \gdef\@thanks{}\gdef\@author{}\gdef\@title{}\let\thanks\relax}

\def\@maketitle{\newpage
 \vbox to 1.85in{
 \vspace*{\fill}
 \vskip 0.8em
 \begin{center}%
  {\Large\bf \@title \par}%
  \vskip 0.20in minus 0.10in
  {\large
   \lineskip .5em
   \begin{tabular}[t]{c}\@author
   \end{tabular}\par}%
 \end{center}%
 \par
 \vspace*{\fill}
 }
}

\def\abstract{\begin{center}%
{\large\bf \abstractname\vspace{-.5em}\vspace{\z@}}%
\end{center}}

\def\section{\@startsection {section}{1}{\z@}{-2.2ex plus-0.6ex minus
    -.2ex}{1.2ex plus.2ex}{\reset@font\large\bf}}

\def\subsection{\@startsection{subsection}{2}{\z@}{-2.0ex plus -0.5ex minus
   -.2ex}{1.0ex plus .2ex}{\reset@font\normalsize\bf}}

\def\subsubsection{\@startsection{subsubsection}{3}{\z@}{-2.0ex plus
 -0.5ex minus -.2ex}{1.0ex plus .2ex}{\reset@font\normalsize\it}}

\makeatother
\usepackage{amsmath,amssymb,amsthm,graphicx,booktabs,tabularx,array,tikz,float}
\usetikzlibrary{arrows.meta,positioning,calc}
\microtypesetup{spacing=false}
\newtheorem{proposition}{Proposition}
\definecolor{kfblue}{HTML}{245B83}
\definecolor{kfteal}{HTML}{246A60}
\definecolor{kfamber}{HTML}{8A601F}
\hypersetup{pdftitle={KedgeFlow: A Programmable Control Plane for Enterprise AI Agents},pdfauthor={Jun He; Deying Yu}}
\begin{document}
\raggedbottom
\title{\bf KedgeFlow: A Programmable Control Plane\\for Enterprise AI Agents}
\author{{\rm Jun He}\\OpenKedge.io \and {\rm Deying Yu}\\OpenKedge.io}
\maketitle
\begin{abstract}
Enterprise AI agents can execute altered versions of approved actions, concurrently reuse the same permission, or retry operations whose responses were lost. These failures arise when approval and resource execution lack a shared contract. KedgeFlow specifies language-independent contracts for trusted controllers to authorize state-changing operations and privileged gateways to enforce that authority independently of agent reasoning. Inspired by OpenFlow, it connects policy checks supported by evidence, signed permission for an exact action, a shared allowance that can be consumed only once, target-side conditions for applying the change, and retained outcome records. Under trusted issuance of exact-action authority, two conditional properties establish at most one successful claim per allowance and link every governed dispatch to a matching grant, claim, and subsequent dispatch authorization, even when the agent is compromised. A reference implementation exposes this execution path to agent systems through the Model Context Protocol (MCP). In the single-resource evaluation, 6,560 concurrent ledger claim attempts yielded exactly one successful claim for each of 82 allowances. Boundary tests rejected unauthorized and altered requests before consumption. A 32-call gateway race produced one committed mock effect, and a retry after a lost response produced no additional dispatch. Deliberately permissive policy nevertheless authorized an undesirable effect, exposing the limits of execution authorization. Evaluation uses three members of the etcd coordinated key-value store on one host, synthetic authentication and authorization inputs, and a mock target that rejects stale requests. Distributed reservations, enforcement of every production access route, native crash recovery, exactly-once effects at real targets, and interoperability between independent implementations remain untested.
\end{abstract}

\section{Introduction}
\label{sec:intro}
\begin{figure*}[t]
\centering
\begin{tikzpicture}[font=\fontsize{9}{10.5}\selectfont,>=Latex,
 card/.style={draw=black!70,rounded corners=2.5pt,fill=white,align=center,inner sep=5pt},
 action/.style={->,semithick},
 control/.style={->,very thick,color=kfblue!90!black},
 shared/.style={<->,very thick,color=kfteal!90!black},
 report/.style={->,densely dashed,semithick,color=black!70},
 query/.style={<->,densely dashed,semithick,color=black!70},
 lab/.style={font=\fontsize{8.5}{10}\selectfont,inner sep=1pt,align=center},
 header/.style={rounded corners=2pt,font=\small\bfseries,inner sep=3pt,align=center}]

\filldraw[draw=kfamber!70,fill=kfamber!6,rounded corners=4pt,thick] (0,-0.7) rectangle (3.95,7.95);
\filldraw[draw=kfblue!70,fill=kfblue!5,rounded corners=4pt,thick] (4.7,-0.7) rectangle (11.3,7.95);
\filldraw[draw=kfteal!80!black,fill=kfteal!6,rounded corners=4pt,thick] (12.2,-0.7) rectangle (17.2,7.95);
\node[header,fill=kfamber!25,draw=kfamber!80,text width=3.3cm] at (1.975,7.55) {Agent Plane\\\scriptsize Untrusted};
\node[header,fill=kfblue!20,draw=kfblue!80,text width=5.6cm] at (8,7.55) {Control Plane\\\scriptsize Trusted authority};
\node[header,fill=kfteal!25,draw=kfteal!90!black,text width=3.8cm] at (14.7,7.55) {Enforcement Plane\\\scriptsize Privileged boundary};

\node[card,text width=3.3cm,minimum height=0.9cm] (runtime) at (1.975,6.5) {\textbf{Reasoning Runtime}\\Untrusted context / LLM};
\node[card,text width=3.3cm,minimum height=1.25cm,draw=kfamber!90!black,thick] (agent) at (1.975,5) {\textbf{Agent Adapter}\\Freeze proposal $q$};
\node[card,text width=3.3cm,minimum height=1cm,dashed,draw=kfamber!80!black,fill=kfamber!10] at (1.975,3.2) {\textbf{Compromised Intent}\\Prompt injection / replay\\Untrusted proposals};
\node[card,text width=3.3cm,minimum height=0.9cm,fill=black!5,draw=black!50] at (1.975,1.4) {\textbf{Credential Isolation}\\No target mutation\\credentials};

\node[card,text width=5.6cm,minimum height=1.1cm,fill=kfblue!8] (policy) at (8,6.5) {\textbf{Policy Authority} / \textbf{Resource Registry}\\Policies, action descriptions, keys};
\node[card,text width=5.6cm,minimum height=1.25cm] (admission) at (8,5) {\textbf{Admission Controller}\\\textbf{Evidence Providers}\\Check policy against supporting evidence};
\node[card,text width=5.6cm,minimum height=0.95cm,draw=kfblue!90!black,thick] (grant) at (8,3.2) {\textbf{Grant Authority}\\Issue exact grant $G$ bound to root $\rho$};
\node[card,text width=4.25cm,minimum height=0.9cm] (outcomes) at (8,1.4) {\textbf{Outcome Service}\\Receipts / outcome correlation};
\node[card,text width=5.6cm,minimum height=0.9cm,fill=kfblue!12,draw=kfblue!90!black,line width=0.9pt] (state) at (8,0) {\textbf{Shared Authority Store} (etcd)\\Roots / versions / attempts};

\node[card,text width=3.8cm,minimum height=1.25cm,line width=1.1pt,draw=kfteal!90!black] (gateway) at (14.7,5) {\textbf{Execution Gateway}\\Authenticate; verify exact $G$\\Claim; authorize dispatch};
\node[card,text width=3.8cm,minimum height=1.1cm,draw=kfteal!80!black] (adapter) at (14.7,3.2) {\textbf{Resource Adapter}\\Fixed mapping / guards\\Target credential custody};
\node[card,text width=3.8cm,minimum height=0.95cm,fill=black!4] (target) at (14.7,1.4) {\textbf{Protected Target Resource}\\Declared guard / commit / query};

\draw[action] (runtime.south) -- (agent.north);
\draw[action,thick] ([yshift=0.3cm]agent.east) -- ([yshift=0.3cm]admission.west);
\node[lab] at (4.4,5.55) {Proposal};
\draw[action,thick] (grant.west) -- (4.4,3.2) -- (4.4,4.7) -- ([yshift=-0.3cm]agent.east);
\node[lab] at (4.4,4.93) {Grant $G$};
\draw[action] (policy.south) -- (admission.north);
\draw[action] (admission.south) -- (grant.north);
\node[lab,anchor=west] at (8.25,4) {Certificate};

\draw[control] (policy.east) -- (13.7,6.5) -- ([xshift=-1cm]gateway.north);
\node[lab,text=kfblue!90!black,font=\fontsize{9}{10.5}\selectfont\bfseries] at (14.4,6.85) {Enforcement Control};

\draw[action,very thick] (agent.west) -- (-0.18,5) -- (-0.18,8.45) -- (17.45,8.45) -- (17.45,5.3) -- ([yshift=0.3cm]gateway.east);
\node[lab,font=\fontsize{9}{10.5}\selectfont\bfseries] at (8.6,8.7) {Execution: exact request + grant $G$};
\draw[action,thick,kfteal!90!black] (gateway.south) -- (adapter.north);
\node[lab,anchor=west,text=kfteal!90!black] at (14.95,4) {Verified call};
\draw[action,thick,kfteal!90!black] (adapter.south) -- (target.north);
\node[lab,anchor=west,text=kfteal!90!black] at (14.95,2.25) {Native call};

\draw[action,kfblue!80!black] (grant.south -| {5.5,0}) -- (state.north -| {5.5,0});
\node[lab,anchor=west,text=kfblue!90!black] at (5.7,0.7) {Install root $\rho$};
\draw[action,kfblue!80!black] (outcomes.south) -- (state.north);
\node[lab,anchor=west,text=kfblue!90!black] at (8.25,0.7) {Retain outcome};

\draw[report] ([yshift=-0.3cm]gateway.west) -- (11.7,4.7) -- (11.7,2.2) -- (9.6,2.2) -- ([xshift=1.6cm]outcomes.north);
\node[lab,anchor=east] at (10.95,2.47) {Receipt / outcome};
\draw[query] (outcomes.east) -- (target.west);
\node[lab] at (11.7,1.65) {Query};

\draw[shared] ([yshift=-0.3cm]gateway.east) -- (17.45,4.7) -- (17.45,-1.35) -- (8,-1.35) -- (state.south);
\node[lab,text=kfteal!90!black,font=\fontsize{9}{10.5}\selectfont\bfseries] at (12.7,-1.07) {Shared claim / dispatch authorization};

\draw[action] (0.4,-1.9) -- (0.9,-1.9);
\node[lab,anchor=west] at (1,-1.9) {Request / grant};
\draw[control] (4.8,-1.9) -- (5.3,-1.9);
\node[lab,anchor=west] at (5.4,-1.9) {Enforcement Control};
\draw[shared] (9.7,-1.9) -- (10.2,-1.9);
\node[lab,anchor=west] at (10.3,-1.9) {Shared state};
\draw[query] (13.4,-1.9) -- (13.9,-1.9);
\node[lab,anchor=west] at (14,-1.9) {Report / query};
\end{tikzpicture}
\caption{KedgeFlow's controller--gateway architecture. The agent submits proposal $q$ and presents signed grant $G$; the control plane owns admission and programs enforcement state. The gateway and adapter hold target credentials and enforce the exact request. Root $\rho$ identifies the shared single-use allowance. The thick Enforcement Control edge programs the gateway; the agent's Execution path consumes authority separately. The authority store retains permission-use state, while the Outcome Service correlates results with attempts. All nine logical roles are represented and may share processes. The reference implementation exercises the gateway, shared store, credential broker, and fixed adapter with a mock target that atomically checks its guard; admission and identity use scripted test components (fixtures), and outcome query is a separately tested library path.}
\label{fig:architecture}
\end{figure*}
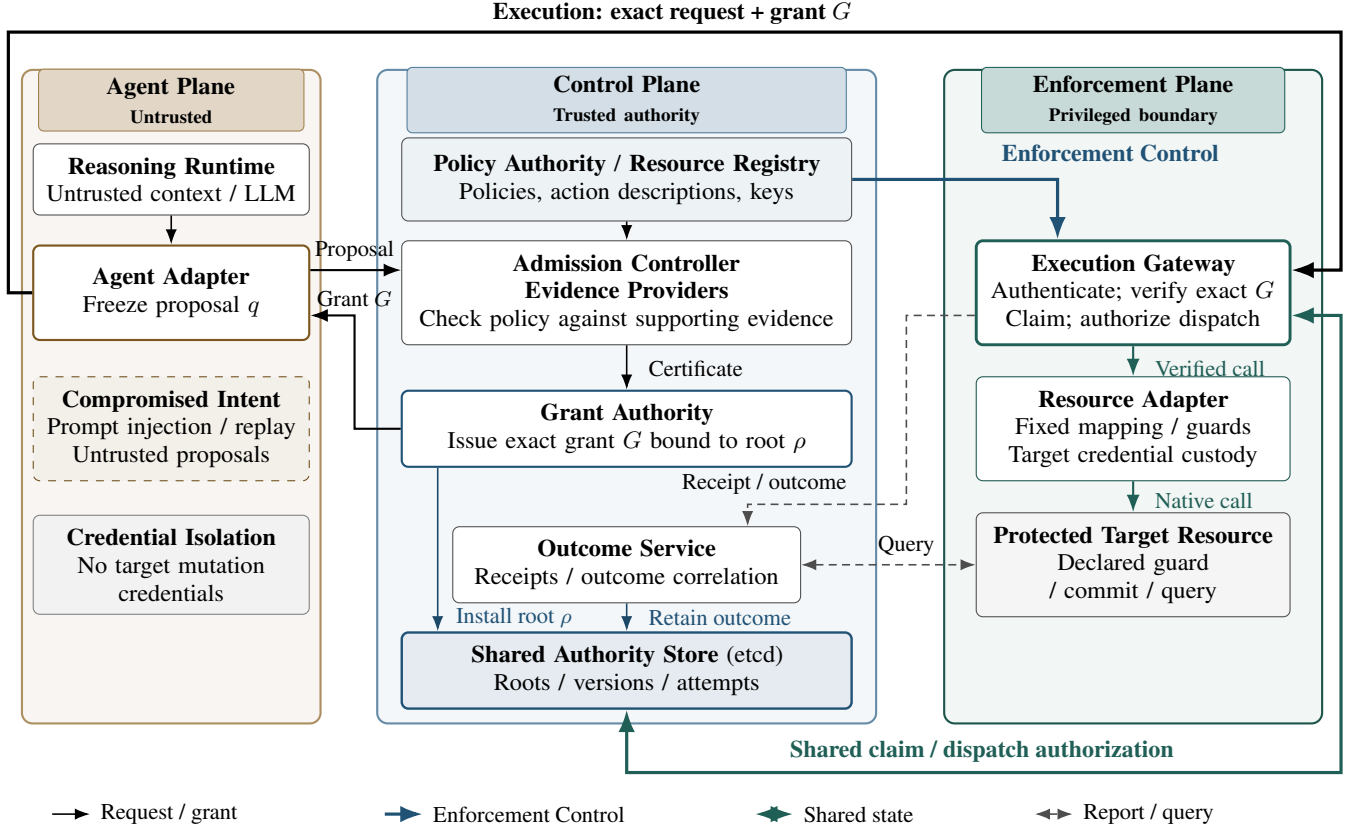

Large language model (LLM) agents can automate infrastructure work, including software development~\cite{devin2024}, cloud administration, and database operations. Tool protocols such as the Model Context Protocol (MCP) connect agents to external services~\cite{mcp2024}. Some tools can change production state: for example, draining a cluster node removes it from active service, while changing a routing rule redirects traffic. Giving an agent credentials for these operations allows its reasoning to determine which changes reach the resource.

An infrastructure agent may base its plan on incident tickets, logs, or tool responses that contain untrusted text. In \emph{indirect prompt injection}, an attacker places instructions in that text to redirect the agent's behavior~\cite{greshake2023}. Input filtering and prompt guardrails can reduce this risk, but they do not independently authorize a resource operation~\cite{owaspprompt,owasptop10}. Caller authentication establishes who sent a request; it does not establish whether that request was approved or whether its parameters still match the approval. If the agent holds credentials that permit production changes, an authenticated but compromised agent can use them to perform harmful actions.

The gap between approval and execution appears in three failure modes. In \emph{parameter substitution}, an agent approved to drain a test node presents the valid permission with a production database node substituted as the target. In \emph{concurrent reuse}, two agent tasks present the same permission simultaneously; separate gateway checks can both succeed unless they share a record of its use. In \emph{lost-response retries}, a timeout leaves the agent unable to determine whether a cloud operation took effect. Repeating a non-idempotent operation---one that can produce an additional effect each time it runs---may then duplicate the change. These cases require checks on the exact request, coordination of permission use, and explicit handling of uncertain outcomes.

KedgeFlow specifies a shared contract for those responsibilities. \emph{Admission} checks a proposed action against policy and supporting evidence independently of the agent. Successful admission permits issuance of a signed \emph{grant}: permission for the specified action, target, and parameters. The grant refers to an \emph{execution root}, the identity of a single-use allowance shared by that grant and any replacement grants. A gateway \emph{claims} the allowance by atomically recording its consumption in shared state, so two concurrent requests cannot both claim it. The gateway then rechecks current authority and authorizes \emph{dispatch}, the act of sending the verified request to the target. Dispatch is distinct from \emph{commitment}, when the target applies the change. A target \emph{guard} is a condition that must hold for the change, such as an expected resource version. If the outcome is uncertain, the allowance remains consumed; losing a response does not create permission to retry.

OpenFlow separates the controller that programs network rules from the switches that enforce them~\cite{openflow151}. KedgeFlow applies this separation across three groups of responsibilities, or \emph{planes}: the untrusted \textbf{Agent Plane} proposes actions; the trusted \textbf{Control Plane} owns admission, authority, and outcomes; and the privileged \textbf{Enforcement Plane} validates requests and applies resource-specific checks (Figure~\ref{fig:architecture}). Under trusted issuance, coordinated authority state, and enforcement on every protected route, a compromised agent cannot change approved parameters or claim an allowance twice. Target guarantees depend on how the declared guard is enforced. Existing policy engines such as Cedar~\cite{cedarauth2026} and credential-exchange mechanisms such as RFC~8693~\cite{rfc8693} can supply components; KedgeFlow specifies how their decisions and state remain consistent across approval and execution.

This paper makes three contributions:
\begin{enumerate}
\item \textbf{A programmable controller--gateway contract.} Nine logical roles and six interfaces assign responsibility for admission, grants, shared allowance, resource enforcement, and outcomes. The contract specifies controller-installed state and gateway feedback (Sections~\ref{sec:architecture}--\ref{sec:interfaces}).
\item \textbf{Two conditional safety properties.} Under explicit trust and shared-state assumptions, at most one attempt claims each allowance, and each governed dispatch matches its grant, claim, and current authorization. These properties distinguish execution authority from target commitment and policy correctness (Section~\ref{sec:execution}).
\item \textbf{A reference protocol implementation for agent integration.} An MCP-facing gateway accepts agent tool calls with signed grants while holding target credentials at the enforcement boundary. It implements exact-request validation, shared single-use claims, credential exchange, guarded dispatch, and outcome retention for a single-resource path. Official MCP client integration, concurrency tests, and adversarial tests validate these mechanisms and their limits (Section~\ref{sec:evaluation}).
\end{enumerate}

Section~\ref{sec:architecture} defines the threat model and architecture. Section~\ref{sec:interfaces} specifies the interfaces and candidate message formats. Section~\ref{sec:execution} develops the two conditional properties and the treatment of uncertain outcomes. Section~\ref{sec:evaluation} reports the reference evaluation, and Section~\ref{sec:related} discusses deployment limits and related work.

\section{Threat Model and Architecture}
\label{sec:architecture}
\subsection{Compromised intent and trusted authority}
The attacker controls text in retrieved runbooks, documentation, tickets, or tool responses and may induce any agent-generated consequential proposal through prompt injection~\cite{greshake2023,owaspprompt}. We conservatively allow the attacker to control the Agent Adapter's proposals, parameters, request identifiers, and presented grants, and to replay previously obtained authority. An authenticated agent can therefore remain malicious: authenticated workload identity identifies the calling process, but does not establish trustworthy intent. Agent-supplied approvals (such as \texttt{approved=true} in an LLM tool response) do not constitute independent authorization. In a qualifying enterprise deployment, the agent runtime holds no credentials for directly changing the protected target. The adversarial experiments use a scripted test component that deterministically emits the redirected request; they measure the execution boundary, not an LLM attack-success rate. KedgeFlow confines the effect of compromised intent under stated assumptions rather than claiming to eliminate prompt injection.

The \emph{trusted computing base} comprises the components whose integrity the guarantees require: policy administration, accepted issuer keys, admission and grant issuance, the coordinated authority store, gateway, credential broker, resource mapping, and the relevant target-side checks and mutation operations. The credential broker exchanges caller credentials for credentials restricted to the target. The attacker cannot modify these components, steal their credentials, roll back retained authority, or bypass an enforced resource boundary. Evidence sources are trusted only for their explicitly accepted statements; a signed source can still be wrong. Admission must distinguish untrusted text from independently authenticated approval. Model compromise and control-plane compromise are different cases: compromise of an authorized issuer or policy administrator defeats the corresponding assumption. A permissive policy or compromised issuer can still authorize harmful effects.

\emph{Complete mediation} requires every route to the protected effect to enforce the boundary. Direct application programming interfaces (APIs), shell credentials, remote sessions, subagents, administrative clients, and inherited tokens must be inventoried and restricted. An agent-side hook or a gateway listening only on the local host does not close those routes. The local bypass experiment tests three specific target routes; operating-system credential isolation and an inventory of all clients able to modify production resources remain deployment assumptions.

\subsection{Three planes and nine roles}
The components must preserve a consistent authorization decision across their handoffs. A policy decision, a signed grant, a gateway check, and a target guard each establish a local property. Their integration must preserve the relationships between those properties: who admitted the final request, which allowance it consumes, which target conditions apply, and what is known about its effect. Correct local components do not establish those relationships when their ownership, identifiers, or state transitions disagree.

Consider two infrastructure agents approved to drain different nodes of a healthy three-node service that requires two available nodes. Each approval is reasonable against the initial state, but concurrent execution violates the availability requirement. This example identifies a coordination requirement for future multi-resource extensions; the present protocol and its safety results cover independently governed single-root actions, not globally serializable reservations across distinct execution roots. Exact-action binding prevents either agent from changing its approved target or expected version. A shared execution root prevents reuse of either agent's allowance, but does not coordinate the two distinct allowances. Shared resource reservations are needed to coordinate those competing operations. The target must check the availability condition atomically with the change to prevent a race between checking and applying it. Finally, \emph{outcome reconciliation} correlates an uncertain attempt with later observations, so an unconfirmed effect is not treated as absent before reserved capacity is released. The reference experiments demonstrate exact binding, atomic root consumption, and retained uncertainty for a single-resource path. They do not demonstrate distributed multi-resource reservations or a fleet-wide availability guarantee; this scenario motivates the broader architecture.

KedgeFlow specifies a programmable controller--gateway boundary for these relationships. Existing runtimes retain reasoning and workflow logic. Trusted controllers admit a frozen request, whose approved fields cannot change, and issue authority for that request. A separately privileged gateway validates the authority, claims its shared allowance, and dispatches through a resource adapter. The execution root identifies one allowance shared across successive grant versions, called \emph{generations}. Retrying a request or replacing an unclaimed grant does not create another allowance.

The \textbf{Agent Plane} owns reasoning, planning, and workflows. Its \emph{Agent Adapter} records the action and parameters of a protected tool call as a proposal and tracks that proposal while admission or execution is pending. The adapter may be a runtime hook, a software development kit (SDK) wrapper, a workflow step, or a protocol client. Its security role ends where the agent can alter its behavior; resource-side enforcement remains necessary.

The \textbf{Control Plane} owns policy, admission, authority, and outcomes. The \emph{Policy Authority} activates versioned rules and revocations. The \emph{Resource Registry} supplies action schemas, which define allowed actions and parameters; resource incarnations, which distinguish a recreated resource from an earlier resource with the same name; and adapter descriptors, which state the supported execution behavior. The \emph{Admission Controller} checks policy requirements against statements from \emph{Evidence Providers}. Providers may supply deterministic checks, telemetry, human review, or model judgments; an evidence statement does not itself authorize execution.

The \emph{Grant Authority} converts an admitted request into exact-action authority and preserves its allowance across replacement grants. The \emph{Outcome Service} retains \emph{execution lineage}, the links between a request, its grant, its claimed attempt, and its outcome, and correlates observations of uncertain attempts. Issuance, enforcement, and outcomes share authority state, but remain distinct responsibilities.

The \textbf{Enforcement Plane} contains the \emph{Execution Gateway} and \emph{Resource Adapter}. The gateway authenticates the caller, checks that the request matches its grant, claims an allowance, and obtains a target credential. The adapter translates the verified action into a specific target API call. It declares the supported guard, \emph{fencing} mechanism for rejecting requests with stale sequence numbers, commit point when the change takes effect, and query behavior for inspecting outcomes. Only this privileged boundary may use target mutation credentials within the protected scope.

Roles may share processes or span them. Roles in one process share the risk of that process being compromised. Existing infrastructure may supply these responsibilities, but the components must agree on execution behavior and preserve credential separation; correctly formatted messages alone are insufficient.

Across these three planes, the execution lifecycle for the node-drain scenario proceeds through three phases. During \emph{admission and issuance}, the Agent Adapter freezes the target node and expected resource version into proposal $q$. The registry supplies the action descriptor; admission checks independently scoped evidence; and the Grant Authority issues a signed grant $G$ bound to that request. These control-plane steps are represented by fixtures in the reference evaluation.

During \emph{validation and dispatch}, the gateway authenticates the caller and checks $G$ against the exact request. It then claims the shared execution root, exchanges for a target-only credential, and authorizes dispatch against current policy and ownership state, in that order. The Resource Adapter maps the verified action to a fixed native call. This sequence executes in the reference implementation and links the caller's request to the same consumed allowance.

During \emph{outcome retention}, the gateway attempts to persist the validated target result or the status \emph{Unknown}, meaning that it cannot establish whether the change took effect. The allowance remains consumed. If a response is lost, a subsequent query seeks an observation correlated with that attempt; it does not redispatch the operation or create fresh authority. The reconciliation library is tested separately with scripted observations, rather than a deployed target-query service.

\subsection{Descriptors state the supported guarantee}
A versioned \emph{Consequential Tool Descriptor} records the action and parameter schemas, resource identity rules, and adapter contract. It also records the consequence tier (the policy's impact classification), conflict domain (the resources over which actions may interfere), and compensation availability (whether a separately authorized action can reverse or offset the effect). Admission and issuance refer to the same descriptor version; an update cannot silently change an already authorized action.

Impact and enforcement strength are distinct: reads can disclose protected data, and idempotent writes---whose repetition adds no further effect---can still be consequential. Policy sets requirements; descriptors state the target mechanisms that can enforce them. An operation identifier alone does not make a condition check atomic with a mutation, and a shell tool may not restrict execution to a fixed set of actions. Tool discovery cannot lower policy requirements. A declared conflict domain identifies operations that admission may need to coordinate. A cryptographic \emph{digest}, a hash value identifying the exact descriptor, binds its interpretation to the grant so a substituted descriptor fails validation. Compensation metadata records supported recovery actions; it does not itself establish that reversal is safe or implemented.

\section{Programmable Interface Contract}
\label{sec:interfaces}
\subsection{Controller--gateway programming contract}
KedgeFlow specifies roles, message formats, validation rules, and observable state transitions independently of implementation language. Implementations may use different languages and internal components, provided they preserve the specified request bindings, canonical record bytes, authority coordination, and outcome semantics. The safety propositions depend on these properties, not on a particular language or storage product. The evaluated reference uses etcd for coordinated authority state; another store would need to provide the same atomicity and durability assumptions. These contracts support independent implementations, but interoperability between such implementations has not been demonstrated.

The \textbf{Enforcement Control Interface} defines how a trusted controller programs the gateway, analogous in architectural role to OpenFlow's controller-to-switch interface~\cite{openflow151}. The controller distributes policy versions; accepted issuer keys and revocations; resource and action descriptors; execution grants and active generations; and ownership, fencing, and supported protocol configurations. Policy and ownership \emph{epochs} are revision numbers used to reject stale authority. In return, the gateway reports supported actions and adapter guarantees, accepted configuration and authority versions, attempt identities with execution receipts, and observed outcomes or explicit uncertainty. A receipt records the reported result of an attempt. Installation must preserve signed constraints and return the accepted version and the scope it covers, or a rejection; stale enforcement state cannot silently authorize affected actions.

Accepted versions and capabilities return through Enforcement Control; attempt receipts, status, and uncertainty return through Execution and Outcome. These exchanges have distinct owners and purposes: controller programming changes accepted authority, while the agent's separate \textbf{Execution Interface} is used by agents to redeem authority for an exact request. The six interfaces are semantic contracts, not six required network protocols.

The reference implementation loads issuer keys and descriptors from operator startup configuration. Trusted policy-publication transactions install policy identity, digest, epoch, and ownership metadata; issuance installs a grant digest and unclaimed root; replacement advances that root's active generation and fence sequence number. These are shared-store library mechanisms, not a policy engine or a controller message protocol. Calls to etcd provide storage operations; they do not establish KedgeFlow interoperability between independent implementations. Signed grants arrive on Execution calls and are verified there. Execution discovery advertises supported message formats and header limits, and tool listing exposes action schemas; neither acknowledges controller installation nor reports the full adapter-guarantee contract. Dynamic installation, control acknowledgments, revocation-record exchange, signed gateway receipts, and remote status services remain future protocol requirements. Table~\ref{tab:interfaces} states the specified operations, reference implementation mechanisms, and outstanding work.

Each authority scope requires one coordinated owner. Separate issuer and gateway stores are insufficient even if each store orders its own operations consistently: issuance, replacement, and claim could succeed against different copies of a root. These transitions must share one atomic authority domain, in which their checks and updates have a common order, or use an equivalent coordination protocol. An etcd transaction orders authority metadata; it does not atomically commit a remote resource effect.

\begin{table*}[t]
\centering\footnotesize
\begin{tabularx}{\textwidth}{@{}p{2.35cm}>{\raggedright\arraybackslash}p{4.25cm}>{\raggedright\arraybackslash}p{5.1cm}>{\raggedright\arraybackslash}X@{}}
\toprule
\textbf{Contract operation} & \textbf{Required semantics} & \textbf{Reference implementation} & \textbf{Outstanding work}\\
\midrule
Governance: publish & Authorized versioned policy, descriptors, keys, ownership and revocations. & Startup keys/descriptors; atomic conditional policy/ownership updates. & Administration service, dynamic installation and per-grant revocation.\\
Runtime Admission: propose / grant & Freeze proposal; obtain certificate, denial or request for evidence; issue bounded authority. & Record/schema vectors, remediation and independent test-policy fixture. & Admission service, evidence evaluation and certificate-to-grant validation.\\
Enforcement Control: install & Authenticate programmed objects and expected version; preserve signed scope. & Static configuration; policy identity/digest/epoch and ownership updates over etcd. & Controller transport, dynamic keys/descriptors/profiles and revocation exchange.\\
Enforcement Control: activate authority & Coordinate exact grants, active root generations and fences in one authority domain. & Tested issuance/replacement transactions; installed digest/root state. Envelope verified on Execution. & Authenticated controller grant distribution and independent peer exchange.\\
Enforcement Control: feedback & Report accepted versions/coverage or rejection and supported adapter guarantees. & Library success to caller; Execution discovery and tool listing only. & Gateway control acknowledgments and full adapter-capability feedback.\\
Execution: execute / status & Authenticate exact request/grant; claim then authorize dispatch; return attempt/receipt. & Bounded JSON transport, grant transmission, gateway, fixed adapter, validated mock results; one SDK path. & Signed gateway receipts, remote status and full transport conformance.\\
Evidence: evaluate / observe & Scoped authenticated statements with provenance, validity and explicit uncertainty. & Schemas/signature vectors and scripted evidence inputs. & Provider service, producer-scope validation and admission semantics.\\
Outcome: report / query & Retain receipts and correlated observations; query never restores allowance or dispatches. & Receipt validation, completion recording, status lookup and reconciliation; scripted query tests. & Outcome service, authenticated native query and stranded-claim recovery.\\
\bottomrule
\end{tabularx}
\caption{Operations within the six specified semantic interfaces. Three rows expand the principal Enforcement Control boundary. Required semantics are candidate contracts; the implementation column identifies implemented mechanisms, and the final column identifies future wire/service requirements. Existing etcd transport is not a standardized KedgeFlow protocol or evidence of independent interoperability.}
\label{tab:interfaces}
\end{table*}

\subsection{Records and identities}
A proposal $q$ fixes the calling user or service (subject $u$), calling process identity (workload $w$), delegation, canonical action, exact parameters $\theta$, resource identity and incarnation $r$, and referenced context. Admission evaluates policy requirements $\Omega(q,s)$ for proposal $q$ in state $s$, under policy version $\Pi_v$. It accepts an evidence manifest $W$, the collection of supporting statements. With $H$ denoting a cryptographic hash, the admission certificate binds $H(q)$, the policy version, $H(W)$, satisfied requirements, validity, and admitted constraints $\kappa$. Missing evidence yields a request for further evidence and a violated requirement yields denial; neither issues execution authority.

The Grant Authority validates the certificate and issues $G$ with three groups of bindings, fields that tie authority to the approved request and its permitted use. Request bindings identify the proposal and certificate, subject and workload, action and parameters, intended gateway (the \emph{audience}), resource incarnation, descriptor, and guard. Authority bindings specify lifetime, policy, and ownership epoch. Consumption bindings identify the execution root, generation, nonce (a one-use marker checked during claim), fence, and allowance. Issuance may narrow approved constraints but cannot broaden them.

This is the architecture's issuance contract. The reference verifier authenticates grant fields and executable bindings; it does not re-evaluate admission or reconstruct the evidence chain from opaque certificate digests.

Identity has three scopes. The \emph{action identity} is the frozen proposal digest. The \emph{root identity} $\rho$ names a conserved execution allowance across replacement generations $\gamma$. The \emph{attempt identity} $t$ names the successful claim and its retained outcome. Grant identifiers distinguish signed permission messages. Changing a request identifier, grant transmission method, or generation cannot create another allowance for an already consumed root. Conversely, different roots are not automatically recognized as one logical business action: trusted admission and grant issuance must prevent duplicate allowances for that action.

\subsection{Candidate wire profile and conformance limits}
The candidate \emph{wire profile} defines the message format used to exchange records. Its JSON envelope, the outer structure surrounding a payload, permits only \texttt{profile}, \texttt{kind}, \texttt{schema}, \texttt{key\_id}, \texttt{payload}, and \texttt{signature}. The reference implementation applies restricted RFC~8785 \emph{canonicalization}, which gives a record a deterministic byte representation for signing and hashing~\cite{rfc8785}: UTF-16 property ordering, no Unicode normalization, safe-integer tokens, and rejection of duplicate keys, malformed Unicode, floating/exponent tokens, and trailing data. It refuses unknown envelope and grant fields. Passing signature test vectors for other record kinds does not establish their full payload validation or signer-role semantics.

Let $C(x)$ be these canonical bytes, $P$ the fixed prefix identifying the protocol's signing domain, and $\|$ byte concatenation. The Ed25519 signing input is $P\|0\|C(E^{-})$, where $E^{-}$ is envelope $E$ without its signature field and $0$ is one NUL byte. Signatures use unpadded base64url, a URL-safe binary-to-text encoding. Record digests are lowercase hexadecimal SHA-256 of $P\|\texttt{/kind}\|0\|C(x)$. Proposal identity uses $x=\{\texttt{schema},\texttt{payload}\}$; the grant digest $D(G)$ covers the complete signed envelope. Millisecond timestamps use an exclusive expiry boundary: authority is invalid at or after its expiry time. Signatures authenticate statements, not their truth or sufficiency.

The reference MCP transport~\cite{mcp2024} uses size-limited JSON POST requests for discovery, tool listing, and invocation. Client and gateway must explicitly negotiate how the grant is transmitted: \emph{inline}, in exactly one dedicated header containing the base64url envelope, or \emph{by reference}, using its 64-character full-envelope digest. In the latter case, an authenticated resolver at a fixed trusted endpoint returns the envelope. Mixing transmission methods, following redirects, using unsupported profiles, or failing resolution causes execution to be refused. Both methods consume the same ledger root. The implementation limits encoded inline headers to 6,144 bytes and resolved envelopes to 65,536 raw bytes. Exact wire identifiers, schemas, and test vectors are documented in the public artifact.

The evaluated node-drain action takes exactly two arguments: the target node and its expected resource version. The grant binds those parameters, resource incarnation, action schema, guard condition, and adapter contract. The adapter adds the fence, attempt ID, grant digest, and resource incarnation to the native request. A terminal mock receipt binds attempt, grant digest, resource incarnation, fence, status, and resulting resource version. Ordinary tool completion is insufficient evidence of a committed mutation.

\emph{Record conformance} checks canonicalization, signatures, schemas, and bindings. \emph{Enforcement-path conformance} adds credential isolation, shared claims, current authority, adapter semantics, and retained outcomes. \emph{Deployment qualification} adds complete mediation, policy/evidence ownership, writer inventory, real target semantics, and recovery. These explanatory scopes are not adopted certification levels. Tests support only a bounded record/enforcement slice; one SDK path does not establish independent peers or full MCP transport conformance.

\section{From Admission to Effects}
\label{sec:execution}
\subsection{Exact validation, claim, and dispatch}
The gateway verifies the signed grant against enrolled issuer keys, current time, and its audience. The independently authenticated subject, workload, proof-key thumbprint (the hash identifying the caller's proof key), and delegation must match the grant, as must the action, resource incarnation, parameters, and descriptor. A request cannot broaden permission by changing its natural-language description or presenting forged approval fields. The reference path supports delegated user-token verification and Demonstrating Proof of Possession (DPoP), which binds a token to a caller-held key~\cite{rfc9449}, plus mutual Transport Layer Security (mTLS) authentication of the calling workload. The adversarial harness substitutes explicit subject-authentication test components.

For the implemented single-resource path, issuance creates an unclaimed root with one active grant/generation. Replacement advances generation and fence only while that exact root remains unclaimed. Claim compares the active root, grant digest, unused nonce, policy and ownership epoch, then atomically writes a claimed root, nonce consumption, attempt, and dispatch-intent record. All participating gateways use the same etcd authority state~\cite{etcdguarantees2026}.

\begin{proposition}[Single-use authority in the bounded path]
For one execution-root allowance, assume issue, replacement, and claim compare and update the same durable root atomically in one authority domain, consumed state cannot be rolled back, and every protected dispatch requires a successful claim. At most one attempt can claim that root. Replacement grants share that same root allowance and retire earlier unclaimed generations without creating additional execution authority.
\end{proposition}
\begin{proof}
A claim requires the exact active root to be unclaimed and changes it to claimed atomically. Ordered transactions cannot satisfy that precondition twice. Replacement also requires the exact unclaimed root and advances its active generation. If replacement wins, the older claim fails; if claim wins, replacement fails. Restoring earlier unconsumed state invalidates the premise.
\end{proof}

Claim precedes credential exchange. RFC~8693 token exchange~\cite{rfc8693} supplies the reference broker pattern; action and attempt bindings are additional KedgeFlow inputs, not standard OAuth fields. The incoming caller token is never forwarded as the target token. 

After exchange, the gateway checks remaining validity, current authority, and dispatch authorization before invoking the adapter. Policy and ownership are compared at the atomic dispatch-authorization transaction; a subsequent revocation cannot recall an authorized in-flight request. That transaction and the network effect are not atomic together. An exchange failure, expired deadline, or refused dispatch does not refund consumption. A durable \emph{outbox} record retains dispatch intent and known status; its presence does not instruct the gateway to retry the target operation.

Proposition~1 establishes at-most-one successful claim per execution root allowance, not exactly-once mutation on the remote target. Two independently created roots for the same logical business action are not automatically deduplicated; preventing redundant root creation belongs to trusted admission and grant authority ownership. A late or duplicated network request requires target-side idempotency or fencing. Coordinating distinct roots that compete for shared capacity requires an explicit multi-resource reservation protocol, which is outside the single-root execution contract. Coordinated multi-step workflows with compensation, cross-cluster reservations, and global atomic commit across arbitrary targets are outside the implementation.

\subsection{Containment of compromised intent}
The second result connects each governed outbound request to one independently issued exact grant, one successful claim, and one subsequent dispatch authorization for the same attempt. It concerns authorization to initiate a request; target acceptance and commitment require separate conditions.

Let trace $T$ be an ordered sequence of execution events, with $c\prec a$ meaning that event $c$ precedes event $a$. Let $q$ be a frozen proposal, $G$ a signed grant, and $x$ an authenticated caller and its delegation. Define the binding tuple $b=(D(G),\rho,\gamma,t)$: full-envelope digest, authority-scoped root, grant generation, and attempt ID. It is notation, not a new wire record.

$\mathsf{Issued}(q,G)$ means trusted independent issuance binds $H(q)$ and its admitted constraints to $G$. $\mathsf{Exact}(G,q,x,d)$ requires matching subject/workload/proof key/delegation, audience, action/parameters, resource incarnation, and adapter contract for dispatch $d$. The proposal/admission link is an issuance premise; the reference gateway checks executable bindings, not the opaque admission chain.

$\mathsf{Claim}(c,b)$ denotes the atomic claim in Section~4.1, matching active grant ID/generation/fence, installed digest, unused nonce, policy ID/digest/epoch, and ownership epoch. It binds nonce and retained intent to $t$.

$\mathsf{Authorize}(a,b)$ rechecks the same claimed root/digest, nonce-to-$t$ binding, and current policy/ownership, then changes that attempt's intent from claimed-not-dispatched to dispatch-authorized.

$\mathsf{Dispatch}(d,b)$ is an outbound adapter request formed from that same verified call. Target acceptance, commitment, and gateway outcome observation are separate events; dispatch alone implies none of them.

\begin{proposition}[Conditional authorized-dispatch boundary]
Assume trusted independent issuance, authenticated caller/workload/delegation, unforgeable grants and collision-resistant digests, durable shared authority without rollback, and mediation of each governed path by the validation and transitions above. For every governed dispatch $d\in T$, there exist $q,G,x,b,c,a$ with valid signatures and lifetime at the documented gateway checks such that
\[
\begin{aligned}
 &\mathsf{Issued}(q,G)\land\mathsf{Exact}(G,q,x,d)\land\mathsf{Claim}(c,b)\\
 &\land\mathsf{Authorize}(a,b)\land\mathsf{Dispatch}(d,b)\land c\prec a\prec d.
\end{aligned}
\]
The claim and authorization belong to the same conserved execution attempt, even under attacker-controlled proposals and presented grants.
\end{proposition}
\begin{proof}
Before claim, the gateway authenticates $x$, verifies $G$, and checks exact executable bindings. Claim records $b$ in shared root/nonce/attempt state. Proposition~1 makes that successful claim unique for $\rho$. After broker exchange, authorization requires the same digest, generation, nonce-bound attempt, current authority, and unadvanced intent. The gateway retains that verified call through final validity checks; the adapter maps it to a fixed native request carrying attempt ID $t$, digest $D(G)$, resource incarnation $r$, and fence. Trusted issuance supplies the proposal link. Thus the witnesses refer to one bound lineage with $c \prec a \prec d$. Gateway dispatch authorization precedes outbound request initiation; it establishes neither target acceptance, target commitment, nor benign business intent.
\end{proof}

The propositions constrain authority consumption and dispatch lineage, not the number of physical target invocations or committed effects. Suppressing repeated outbound invocations requires gateway execution discipline, while duplicate-effect prevention may additionally require target-side idempotency, fencing, or an atomic guard. These mechanisms must cover the relevant operation and retry behavior; a declared guard alone does not establish exactly-once effects.

The implemented gateway verifier validates executable grant bindings and signatures against enrolled keys; it does not reconstruct an admission proof or re-evaluate evidence requirements from opaque certificate digests. Proposition~2 establishes that compromised agent intent cannot dispatch unauthorized mutations through a governed path without matching independently issued authority and successful shared-state authorization. It does not prove that policy is sound, that evidence is truthful, or that prompt injection is eliminated: a permissive policy or compromised issuer can still authorize harmful actions, and bypass routes around the gateway fall outside the mediated path. The proposition establishes neither policy soundness, exactly-once mutation, recall of remote requests, nor conditions beyond the declared target guard.

\subsection{Target guards and completion}
A Resource Adapter declares one of three enforcement levels for a concrete action and guard scope. \textbf{Level A} atomically checks a declared guard with the protected mutation. \textbf{Level B} revalidates before dispatch but leaves a race between the check and the effect. \textbf{Level C} observes an effect after acceptance. A product cannot be assigned a level independently of its actual action path.

If every client able to modify the resource uses a correct atomic check-and-mutate operation, Level A excludes a mutation whose declared guard is false at the commit point. The reference mock checks resource version and an increasing fence under the same lock as its mutation. This establishes the local operation's scope; it does not establish a fleet-wide health condition or correctness of native database, Kubernetes, or cloud-service adapters.

Completion mode is separate. Dispatch-bound validity constrains initiation. Commit-bound validity additionally requires the target to enforce the relevant deadline or authority condition at commit. Cancellation or ledger revocation cannot retract a request already accepted by a remote target. The reference execution path supports dispatch-bound completion and refuses unsupported stronger profiles.

\subsection{Unknown outcomes retain consumption}
A lost response can mean no effect or an effect whose acknowledgment disappeared. The gateway reports Unknown and attempts to persist it without restoring allowance. Only \emph{persisted Unknown} enters the reconciliation API. A failed completion write or process crash can instead leave a durable \emph{Claimed} attempt: its allowance is consumed, but no completion status has been saved.

Five distinct crash intervals expose this gap: crashing \emph{after claim but before credential exchange} leaves the root consumed and durable intent recorded as claimed-not-dispatched, an intermediate state that does not constitute terminal no-dispatch evidence since a live executor could still advance it; crashing \emph{after exchange but before dispatch authorization} leaves a minted target credential before authorization is recorded in shared state; crashing \emph{after authorization but before target contact} leaves dispatch-authorized intent without confirming whether the outbound request was initiated or reached the target; crashing \emph{after target contact but before observing a validated outcome} leaves target mutation status unestablished; and crashing \emph{after outcome receipt but before durable completion recording} leaves the observed target result unpersisted in the attempt record. The reference implementation has no separate durable receipt journal; these are code-derived windows, not executed crash experiments.

The required recovery discipline follows six core rules: never automatically return consumed authority to the available pool, preserve durable attempt identity across crash and restart boundaries, distinguish ``no dispatch'' from ``no effect'', require authenticated target evidence for committed-effect conclusions before updating status from Unknown to Committed, account for late or delayed effects before reissuing fresh authority or initiating compensation, and avoid blind retries of non-idempotent operations.
The winning attempt is retained in shared state; a replay rejection does not recover its identity. Correlation requires attempt, root, grant digest, and a nonempty target reference; an expected resource value alone is insufficient. The reconciliation library changes Unknown to Committed only for a matching committed observation. Pending or other observations retain Unknown, and its scripted querier tests correlation and conservation only.

Queries never redispatch or replenish authority. Even terminal outcomes recording nondispatch or guard refusal retain consumption. Reissuance or compensation needs authoritative no-effect or committed-effect evidence as appropriate, protection against late effects, fresh admission, and conserved budget. Compensation is a separately authorized action. Stranded-claim repair, executor recovery barriers, and production authenticated target-query reconciliation remain future implementation work; the current implementation does not provide complete crash recovery.

\section{Reference Implementation and Evaluation}
\label{sec:evaluation}
The evaluation checks three properties of the implemented single-resource path: \emph{exact-request validation} rejects changed bindings, invalid authority, or unapproved parameters before consuming allowance; \emph{shared state} preserves single-use allowance under concurrent claims, retries, and uncertain outcomes; and \emph{execution authorization} remains distinct from policy correctness, since valid authority may permit a harmful action.
The retained experiments support the first two boundary properties and expose the third limit through a deliberately permissive policy. They are functional tests of specified mechanisms, rather than estimates of LLM attack-success rates or production performance.

\begin{table*}[t]
\centering\small
\setlength{\tabcolsep}{3.5pt}
\begin{tabularx}{\textwidth}{@{}Xrrrrrrrrl@{}}
\toprule
\textbf{Attacker-controlled case} & \textbf{GW} & \textbf{Probe} & \textbf{Auth.} & \textbf{Grant} & \textbf{Claim} & \textbf{Authz.} & \textbf{Target} & \textbf{Commit} & \textbf{Unknown}\\
\midrule
Injected runbook $\to$ unauthorized mutation & 1 & 0 & 0 & 0 & 0 & 0 & 0 & 0 & No\\
Forged approval in tool response & 1 & 0 & 0 & 0 & 0 & 0 & 0 & 0 & No\\
Changed resource after admission & 1 & 0 & 1 & 1 & 0 & 0 & 0 & 0 & No\\
Changed parameters after admission & 1 & 0 & 1 & 1 & 0 & 0 & 0 & 0 & No\\
Grant payload substitution & 1 & 0 & 1 & 1 & 0 & 0 & 0 & 0 & No\\
Other-audience signed envelope replay & 1 & 0 & 1 & 1 & 0 & 0 & 0 & 0 & No\\
No-grant direct gateway request & 1 & 0 & 0 & 0 & 0 & 0 & 0 & 0 & No\\
Invalid-grant direct gateway request & 1 & 0 & 0 & 0 & 0 & 0 & 0 & 0 & No\\
Concurrent duplicate gateway calls & 32 & 0 & 1 & 1 & 1 & 1 & 1 & 1 & No\\
Lost response, then second induced call & 2 & 0 & 1 & 1 & 1 & 1 & 1 & 1 & Yes\\
Undesirable action, permissive policy & 1 & 0 & 1 & 1 & 1 & 1 & 1 & 1 & No\\
Local direct-target bypass probes & 0 & 3 & 0 & 0 & 0 & 0 & 0 & 0 & No\\
\bottomrule
\end{tabularx}
\caption{Twelve deterministic boundary tests. GW counts gateway calls; Probe counts direct-target HTTP/TLS attempts. Auth. is independent test-policy authorization; Grant is signed issuance; Claim is successful root consumption; Authz. is successful dispatch authorization. Target counts observed mock requests passing its bearer check before guard/mutation; Commit counts mutations. Unknown is gateway-visible uncertainty, including an observed mock commit with a lost response. Rows use fresh authority namespaces/targets. Scripted redirection is assumed, not a measured LLM attack rate.}
\label{tab:adversarial}
\end{table*}

\subsection{Implementation scope and experimental setup}
The reference implementation is written in Rust and exposes the execution path to agent systems through an MCP-facing gateway. Agents retain their planning and workflow logic and submit tool calls with signed grants; the gateway handles exact-request checks, shared claims, credential exchange, adapter dispatch, and outcome retention. Client and gateway negotiate inline or reference grant transmission. The official MCP client integration exercises discovery, tool listing, and protected invocation. The evaluated resource adapter maps calls to a fixed fenced mock target; production target adapters and controller services remain future work (Table~\ref{tab:interfaces}).

The adversarial harness connects a minimal independent authorization decision to signing, issuance, and an in-process gateway call. It uses synthetic subject authentication and workload issuance, with actual outbound mTLS to the broker and target. Other suites exercise cryptographic caller authentication, SDK compatibility, and transport boundaries separately.

The public artifact preserves the evaluated reference runtime and its dependency lockfiles. Snapshot hashes are checked against the upstream source revision; exact source and dependency versions are recorded in the artifact. Historical Go proxy or admission results are excluded. A Go harness launches the official embedded etcd test cluster.

The evaluation runs on one ARM64 macOS workstation. Three etcd members use distinct local endpoints for client requests and Raft consensus, along with separate data directories. Cargo builds and tests the reference runtime offline with locked dependencies; the test configuration requires real etcd and the official MCP client SDK, so missing dependencies fail instead of silently skipping tests. All \textbf{58 implementation tests pass}, with no failures, ignored tests, or detected skips: 52 upstream tests, five paper safety tests, and one test containing 12 adversarial scenarios.

The instrumented chain counts authorization decisions, signed grants, successful claims, dispatch authorizations, credentialed target requests, mock commitments, and gateway-visible uncertainty separately. A claim consumes allowance before target contact. Dispatch-authorization and target-request counts coincide in these scenarios but are measured independently; other suites receive no inferred stage counts.

\subsection{Concurrency, conservation, and recovery}
The root-consumption tests check whether concurrent clients can redeem one allowance twice. The direct-ledger matrix synchronizes 32 or 128 contenders for 1, 8, or 32 fresh roots, totaling \textbf{6,560 claim attempts, 82 successful claims, 6,478 consumed-nonce refusals, and zero transport errors}. Exactly one wins per root. This matrix has no target dispatches and measures neither committed operations nor throughput.

Issuance/replacement tests check whether authority updates add allowance: 32 issuers race for one root (one issuance and 31 refusals), a chain retires 32 unclaimed generations, and terminal-state tests preserve consumption after commitment, uncertainty, known nondispatch, and guard refusal. Reconciliation tests check whether correlated observations can refine an outcome without creating authority: a fresh client reconstructs Unknown, refuses four wrong correlation fields, keeps a pending observation unknown, and accepts a matching committed observation without reopening authority. Client recreation is not an etcd crash, quorum-loss, or restore test.

Existing full-path code combines real etcd, TLS broker and target endpoints, synthetic subject/actor issuance, and the fenced mock. It exercises a committed effect, stale guard refusal, and a committed effect followed by a failed response. Other suites reject signed-envelope, identity, parameter, endpoint, descriptor and receipt substitutions. The official SDK test exercises discovery, tool listing, remediation, fixture execution and replay refusal with explicit synthetic authentication/broker boundaries.

\subsection{Deterministic adversarial boundary tests}
Table~\ref{tab:adversarial} relates each scenario to an execution-contract boundary: missing or forged authority tests independent authorization; changed resource or parameters and other-audience envelopes test exact binding (distinct from single-use replay); concurrent redemption and lost-response replay test authority conservation. Permissive policy separates valid authority from desirable effects, and bypass probes test selected credential routes.

Synthetic runbook text orders the agent to ignore maintenance and drain a protected node; a synthetic tool response asserts approval and a forged administrator override. The compromised-agent fixture emits the consequential request deterministically. A separate test-policy function permits the registered resource only with a trusted-review flag or a deliberately permissive policy; it never interprets tool text as approval. This is a minimal policy harness, not a shipped admission service or model experiment.

Table~\ref{tab:adversarial} reports the measured events across the three categories of boundary tests.
In evaluating \emph{refusal before consumption}, admission in the first two cases refuses the request and issues no grant; subsequent ungranted gateway calls request evidence without claiming authority. For post-admission substitution, legitimate authority is issued but changed parameters fail exact binding. Modified signature payloads and other-audience envelopes are also refused before claim, and direct calls without authority produce no effects.

Under \emph{conservation under contention and uncertainty}, the 32-call gateway race produces one claim, broker exchange, dispatch, commit, and successful response, accompanied by 31 refusals of consumed authority. In the lost-response case, the mock commits once then returns HTTP~503; the gateway retains Unknown, and the second call with a changed request identifier is refused. While the mock observer knows a commit occurred, the gateway cannot establish that outcome from the response. Preserving consumption prevents this gap from becoming new authority, without premature reconciliation.

\begingroup
\widowpenalty=10000
\clubpenalty=10000
The experiments also expose \emph{policy and mediation limits}. Under deliberately permissive policy, an undesirable drain obtains valid authority and commits: exact binding cannot repair a policy that authorizes the wrong action. For local bypass probes, a caller possessing the test mTLS identity receives HTTP~401 without broker credentials; a client without a certificate fails TLS. The target checks a separate broker bearer credential before mutation. These are actual restrictions on tested routes, but fixture code knows the credential; they do not establish operating-system credential isolation, coverage of every modifying client, or protection if that token is stolen.

\begin{samepage}
The evaluation strictly maintains the distinction between protocol records, ledger claims, dispatch authorizations, credentialed target requests, and committed effects. We do not derive throughput, availability, LLM attack-success rates, or production-scale guarantees from these experiments. The public artifact includes raw logs, baseline results, and a machine-readable claim-to-evidence audit; uninstrumented stages remain unmeasured.
\par\end{samepage}
\par\endgroup

\section{Limitations and Related Work}
\label{sec:related}
\subsection{Deployment and evidence limits}
Qualifying a deployment requires identifying every client able to cause each protected effect, the descriptor owner, accepted evidence, issuance scope, credential holder, target-side enforcement operations, and outcome correlation rules. Each added adapter needs checks for alternate routes, request matching, races, guards, and uncertain outcomes. Registry declarations alone do not prove coverage.

Table~\ref{tab:limits} consolidates the principal limits. New dispatch depends on the authority store and broker being available; refusal is an availability cost of the boundary. An already dispatched unknown effect may remain unresolved indefinitely without authoritative target evidence. Revocation prevents dispatch authorization ordered after its activation, but cannot retract an earlier authorization or an accepted remote request. Retention and restoration procedures must preserve consumed authority, not restore an older allowance snapshot.

\begin{table}[H]
\centering\footnotesize
\renewcommand{\arraystretch}{0.86}
\setlength{\tabcolsep}{2pt}
\begin{tabularx}{\columnwidth}{@{}p{1.75cm}X@{}}
\toprule
\textbf{Boundary} & \textbf{Scope or untested assumption}\\
\midrule
Agent attacks & Redirection assumed; no evaluated model or attack corpus.\\
Admission & Minimal fixture; production policy/evidence/issuer assumed.\\
Mediation & Three probes; OS custody and writer coverage untested.\\
Distributed & One host; no crash, quorum-loss, or cross-cluster tests.\\
Target effects & Fenced mock; no native adapter or physical exactly-once.\\
Recovery & Library query fixture; no target status service or crash journal.\\
Protocol & Candidate slice; independent peers and full transport untested.\\
Scope & No adoption, throughput or comparative performance metrics.\\
\bottomrule
\end{tabularx}
\caption{Limits of conditional guarantees and executed evidence.}
\label{tab:limits}
\end{table}

\subsection{Existing infrastructure patterns}
OpenFlow supplies the controller--enforcement pattern~\cite{openflow151}. The National Institute of Standards and Technology's zero-trust architecture separates policy decision and enforcement~\cite{nist800207}; Cedar evaluates requests~\cite{cedarauth2026}; OAuth token exchange (RFC~8693)~\cite{rfc8693} and DPoP (RFC~9449)~\cite{rfc9449} isolate and bind credentials. Capability authorization and least-privilege credential isolation are established mechanisms. KedgeFlow contributes their explicit systems composition, uniting frozen action identity, admission supported by evidence, exact-action grants, single-use allowance, coordinated claim and dispatch authorization, declared adapter guarantees, and retained outcome records under uncertainty. Systems with sufficiently capable policy decision and enforcement components can realize these contracts; no mechanism is claimed as an isolated invention.

Distributed authority and transactional systems provide foundational primitives: stores like etcd supply atomic state transitions~\cite{etcdguarantees2026}, while resource-side conditional writes protect declared target conditions. Separate commit domains do not become one distributed transaction by composition: an etcd transaction orders authority metadata without committing remote target mutations. Target idempotency suppresses duplicate effects given stable operation identity. Durable workflows preserve progress, while activity retries require explicit treatment of non-idempotent side effects~\cite{temporalactivities2026}. KedgeFlow's distinction is the specified ownership and interface composition across runtimes, admission systems, authority services, gateways, and heterogeneous targets.

Enterprise effects require varied target conditions, completion modes, and no-effect evidence beyond a uniform OpenFlow match/action table. ``Open'' denotes public candidate contracts, not an adopted standard or demonstrated interoperability. Future comparisons should provide equivalent policy, evidence, credential isolation, guards, and reconciliation before measuring integration costs; functional equivalence is a valid result.

\subsection{Agent action governance}
Greshake et al. show how retrieved content can redirect LLM applications via indirect prompt injection~\cite{greshake2023}. Input filtering and guardrails address that channel, while OWASP recommends tool permission checks, parameter validation, and least privilege outside the model~\cite{owaspprompt,owasptop10}. Action governance includes policy decisions, action identity, approval binding, and outcome evidence. KedgeFlow focuses on composing those decisions with conserved execution authority and distributed resource semantics, binding exact approved parameters and declared guarantees across concurrency, retries, stale authority, and uncertain outcomes.

Agent Control Protocol studies temporal admission~\cite{fernandezacp2026}; Proof-Carrying Agent Actions establishes runtime-neutral certificates and outcome closure~\cite{pcaa2026}; CAVA ensures canonical identity and receipt integrity~\cite{cava2026}; and OpenAgentFlow normalizes session-aware action governance across agent fleets~\cite{openagentflow2026}. While these systems address action binding and admission, KedgeFlow contributes the connecting controller--gateway authority contract and conserved resource semantics under concurrency and uncertainty without making comparative performance claims.

Related line-of-work papers address evidence obligations~\cite{he2026cac}, assurance-aware scheduling~\cite{he2026aas}, commit conditions~\cite{he2026tct}, and certificate-bound authority~\cite{he2026seb}; these can occupy KedgeFlow roles but are not required as a set. The longer KedgeFlow whitepaper provides extended architectural reference, while the contract and conditional proofs presented here are self-contained.

\section{Conclusion}
\label{sec:conclusion}
KedgeFlow defines a programmable controller--gateway contract for independently authorized, consequential actions in enterprise distributed systems. It separates agent reasoning from mutation authority, specifying how authority is bound, consumed, enforced, and reconciled. Two conditional results connect single-use authority and governed dispatch to explicit shared-state and issuance assumptions. The reference implementation exposes the execution path through MCP for integration with agent systems. Its evaluation demonstrates refusal of altered requests and retained consumption under concurrent calls and lost responses across 58 tests (12 deterministic adversarial scenarios); permissive policy exposes the limit of execution authorization. Complete mediation, sound admission, native recovery, and interoperability between independent implementations remain outside the demonstrated guarantees.

\par\noindent\textbf{Artifact availability.} The reproducibility artifact is publicly available at \url{https://github.com/openkedge/kedgeflow} (under \texttt{artifacts/short-paper} at commit \texttt{769c0a8}), containing the pinned runtime, test suite, fixtures, lockfiles, logs, manifests, and reproduction instructions.

\par\noindent\textbf{AI-use disclosure.} OpenAI Codex assisted with revision, tests, and simulated reviews; authors remain responsible for all text, code, and reported results.

\sloppy
\interlinepenalty=10000 
{\small\let\oldthebibliography\thebibliography\renewcommand\thebibliography[1]{\oldthebibliography{#1}\setlength{\itemsep}{0pt}\setlength{\parskip}{0pt}}\bibliographystyle{unsrt}\bibliography{refs}}
\end{document}